\documentclass[conference]{IEEEtran}
\IEEEoverridecommandlockouts
\usepackage{cite}
\usepackage{amsmath,amssymb,amsfonts}
\usepackage{algorithmic}
\usepackage{graphicx}
\usepackage{textcomp}
\usepackage{xcolor}
\def\BibTeX{{\rm B\kern-.05em{\sc i\kern-.025em b}\kern-.08em
    T\kern-.1667em\lower.7ex\hbox{E}\kern-.125emX}}
\usepackage{mathtools}
\usepackage{amsthm}
\usepackage{stfloats}
\usepackage{epstopdf}
\newtheorem{theorem}{Theorem}

\begin{document}

\title{AFDM: A Full Diversity Next Generation Waveform for High Mobility Communications}

\author{
    \IEEEauthorblockN{Ali Bemani\IEEEauthorrefmark{1}, Nassar Ksairi\IEEEauthorrefmark{2}, and Marios Kountouris\IEEEauthorrefmark{1} }
    \IEEEauthorblockA{\IEEEauthorrefmark{1}Communication Systems Department, EURECOM, Sophia Antipolis, France
    \\\{ali.bemani, marios.kountouris\}@eurecom.fr}
    \IEEEauthorblockA{\IEEEauthorrefmark{2}Mathematical and Algorithmic Sciences Lab, Huawei France R\&D, Paris, France,
    \\ nassar.ksairi@huawei.com}
}

\maketitle
\begin{abstract}
We present Affine Frequency Division Multiplexing (AFDM), a new chirp-based multicarrier transceiver scheme for high mobility communications in next-generation wireless systems. AFDM is based on discrete affine Fourier transform (DAFT), a generalization of discrete Fourier transform characterized with two parameters that can be adapted to better cope with doubly dispersive channels. Based on the derived input-output relation, the DAFT parameters underlying AFDM are set in such a way to avoid that time domain channel paths with distinct delays or Doppler frequency shifts overlap in the DAFT domain. The resulting DAFT domain impulse response thus conveys a full delay-Doppler representation of the channel. We show that AFDM can achieve the full diversity of linear time-varying (LTV) channels. Our analytical results are validated through numerical simulations, which evince that AFDM outperforms state-of-the-art multicarrier schemes in terms of bit error rate (BER) in doubly dispersive channels.
\end{abstract}

\section{Introduction}
Next generation wireless systems (beyond 5G/6G) are envisioned to support a wide spectrum of services and applications, including reliable communications at high carrier frequencies in high mobility environments (e.g., high-speed railway systems, vehicular-to-infrastructure, and vehicular-to-vehicular communications). Various multicarrier techniques, such as orthogonal frequency division multiplexing (OFDM) and single-carrier frequency division multiple access (SC-FDMA), have been deployed in standardized communication systems. These schemes have been shown to achieve satisfactory or even optimal performance in time-invariant frequency selective channels. However, orthogonality among subcarriers is destroyed due to large Doppler frequency shifts in high mobility scenarios, resulting in deteriorated performance. 

The optimal way to cope with time-varying multipath (also called doubly dispersive) channels is to let the information symbols modulate a set of orthogonal eigenfunctions of the channel (input/output relation) and project the received signal over the same set of eigenfunctions at the receiver. However, in contrast to linear time-invariant (LTI) systems in which complex exponentials are eigenfunctions, finding an orthonormal basis for general linear time-varying (LTV) channels is not trivial. Polynomial phase models that generalize complex exponentials are often used as alternative bases. Interestingly, there are special cases for which an orthonormal basis is formed by chirps, i.e., complex exponentials with linearly varying instantaneous frequencies. Despite not being optimal in general, chirp-based techniques can be adjusted to the channel characteristics as a means to achieve near-optimal performance \cite{erseghe2005multicarrier}. 
Using a chirp basis instead of the sine basis for transmission over time-varying channels is first introduced in \cite{martone2001multicarrier}, where fractional Fourier transform (FrFT) is used to generate multi-chirp signals. However, the approximation used for discretizing the continuous-time FrFT leads to imperfect orthogonality among chirp subcarriers and hence to performance degradation. A multicarrier technique based on a specific discretization of the affine Fourier transform (AFT) is proposed in \cite{erseghe2005multicarrier}. Discrete AFT (DAFT) parameters are properly tuned using partial channel state information (CSI), namely the delay-Doppler profile of the channel (known delays and the Doppler shifts of channel paths). That way, the resulting multicarrier waveform, referred to as DAFT-OFDM in the sequel, is equivalent to OFDM with reduced inter-carrier interference (ICI) on doubly dispersive channels. While this property enables low-complexity detection, being equivalent to OFDM implies a diversity order close to or equal to that of OFDM (which is known to be low without channel coding). Moreover, the delay-Doppler profile of the channel is required at the transmitter side to properly tune the DAFT-OFDM parameters. Orthogonal chirp division multiplexing (OCDM) \cite{ouyang2016orthogonal}, which is based on the discrete Fresnel transform - a special case of DAFT, is shown to outperform OFDM in time-dispersive channels thanks to a higher diversity order. However, OCDM cannot achieve full diversity in general LTV channels, since its diversity order depends on the delay-Doppler profile of the channel. Orthogonal time frequency space (OTFS) modulation \cite{hadani2017orthogonal}, a recently proposed waveform for high mobility communications, is a two-dimensional (2D) modulation technique that uses the delay-Doppler domain for multiplexing information. Its diversity order without channel coding is shown to be one \cite{surabhi2019diversity}, however OTFS can be made to achieve full diversity with a phase rotation scheme using transcendental numbers. The concept of \emph{effective} diversity order, i.e., the diversity order in the finite signal-to-noise ratio (SNR) regime, is introduced in \cite{raviteja2019effective}, showing that OTFS may achieve full effective diversity.

In this paper, we take a fresh look at chirp-based multicarrier systems and propose Affine Frequency Division Multiplexing (AFDM), a novel DAFT-based waveform using multi-chirp signals. We first derive the input-output relation, which allows us to adapt the AFDM parameters in a way to avoid that time domain channel paths with distinct delays or Doppler frequency shifts overlap in the DAFT domain. We analytically show that AFDM achieves the full diversity of LTV channels, as opposed to OFDM, DAFT-OFDM, and OCDM. Compared with OTFS, AFDM has comparable performance in terms of BER with lower complexity though. In contrast to the simpler, one-dimensional (1D) transform in AFDM, the 2D transform in OTFS has several drawbacks in terms of pilot overhead and multiuser multiplexing overhead \cite{raviteja2019embedded}. In a nutshell, AFDM is a promising new waveform for high-mobility communications in next generation wireless networks. 
\section{Affine Fourier Transform}
In this section, we introduce the AFT and the DAFT, which form the basis of AFDM. The AFT, also known as the linear canonical transform \cite{LCT} or generalized Fresnel transform, is a continuous transformation that maps a continuous-time signal $s(t)$ into $S(f)$ as follows \cite{pei2001relations}:
\begin{equation}
    S(f) = \begin{cases} \int_{-\infty}^{+\infty}s(t){{\rm e}^{-j({a\over 2b}f^{2}+{1\over b}ft+{d\over 2b}t^{2})}\over \sqrt{2\pi\vert b\vert }}{\rm d}t,& b\neq 0 \\s(df){e^{-j{cd\over 2}f^{2}}\over \sqrt{a}}, \qquad \qquad \qquad \quad \,& b=0\end{cases}\label{AFT}
\end{equation}
provided that parameters ($a,b,c,d$) form  an invertible matrix
$\textrm{M} = \left[{\begin{array}{cc}
a & b \\ c & d
\end{array}}\right]$ with determinant $ad - bc = 1$.
AFT is an integral transformation that generalizes many standard transforms, such as Fourier transform (0,1,-1,0), Laplace transform (0,$j$,$j$,0), and $\theta$-order fractional Fourier transform ($\cos\theta,\sin\theta,-\sin\theta,\cos\theta$). Gauss-Weierstrass, Fresnel, and Bargmann transforms are also special cases. It can be visualized as the action of the special linear group SL$_2(\mathbb{R})$ on the time–frequency plane and is a particular case of a phase space transform, the special affine Fourier transform (SAFT) \cite{SAFT}.

Possible discretizations of the AFT have been discussed in \cite{erseghe2005multicarrier,candan2000discrete, pei2000closed}. In the remainder, we employ the DAFT from \cite{erseghe2005multicarrier}, for which it is shown that the periodicity considerations in Fourier analysis while sampling $s(t)$ and $S(f)$ can be generalized by ensuring that the following constraints hold
\begin{equation}
    s(nT + kT_p)e^{-j2\pi k_1(nT+kT_p)^2} = s(nT)e^{-j2\pi k_1(nT)^2},\label{chrip_per}
\end{equation}
\begin{equation}
    S(nF + kF_p)e^{-j2\pi k_2(nF+kF_p)^2} = S(nF)e^{-j2\pi k_2(nF)^2}
\end{equation}
where $T$ and $F$ are sampling quanta, $T_p$ and $F_p$ are signal periods ($T_p = NT$ and $F_p = NF$), and $N$ is the number of samples.
The relation between $F$ and $T$ is $F = 1/\beta T_p = 1/\beta NT $ and $k_1 = \frac{d}{4\pi b}$, $k_2 = \frac{a}{4\pi b}$ and $\beta =\frac{1}{2\pi b} $. For our purposes, only the constraint \eqref{chrip_per} matters, whose sole practical effect is on the kind of prefix one should add to a DAFT-based multicarrier symbol. Arranging samples $s(nT)$ and $S(nF)$ in the period $[0, N)$ in vectors
\begin{equation}
\begin{aligned}
    \mathbf{s} &= (s_0, s_1, ..., s_{N-1}), \quad s_n = s(nT)\\
    \mathbf{S} &= (S_0, S_1, ..., S_{N-1}), \quad S_n = S(nF),
    \label{eq:chirp_periodicity}
\end{aligned}
\end{equation}
DAFT is expressed as 
\begin{equation}
    \mathbf{S} = \mathbf{As}, \quad \mathbf{A} = \mathbf {\Lambda}_{c_2}{\mathbf F}{\mathbf \Lambda}_{c_1}
\end{equation}
where $\mathbf{F}$ is the discrete Fourier transform (DFT) matrix with entries $e^{-j2\pi mn/N}/\sqrt{N}$ and 
\begin{equation}
    \mathbf{\Lambda}_{c}= {\rm diag} (e^{-j2\pi cn^{2}}, n=0, 1, \, \ldots\,, N-1),
\end{equation}
with $c_1 = k_1T^2$ and $c_2 = k_2F^2$ and
\begin{equation}
    \mathbf{A}^{-1} = \mathbf{A}^H =  {\mathbf \Lambda}_{c_1}^H{\mathbf F}^H{\mathbf \Lambda}_{c_2}^H.
\end{equation}

\section{Affine Frequency Division Multiplexing}
\label{sec:afdm}
In this section, we introduce AFDM, a DAFT-based multicarrier transceiver concept. Inverse DAFT (IDAFT) is used to map data symbols into the time domain and DAFT is performed at the receiver to obtain the effective discrete affine Fourier domain channel response to the transmitted data, as shown in Fig. \ref{fig:AFDM_blokcdiagrma}.
\begin{figure}
  \centering

  \includegraphics[scale=.5]{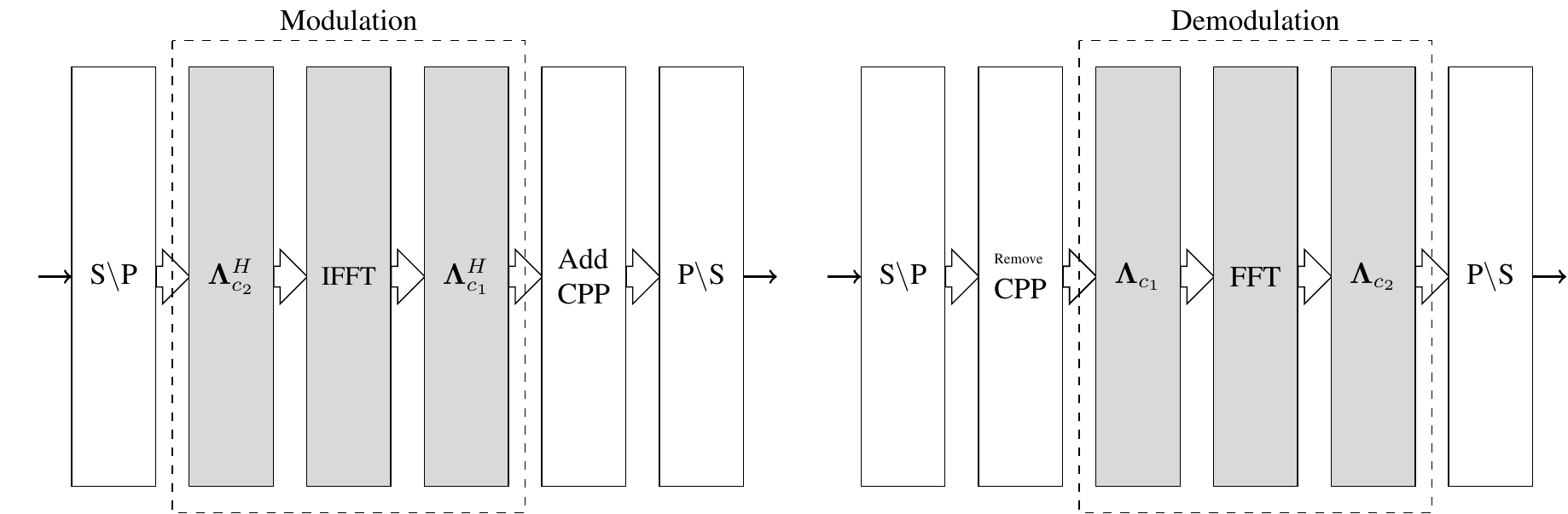}
  \caption{AFDM block diagram}
  \label{fig:AFDM_blokcdiagrma}
\end{figure}
\subsection{Modulation}
Let $\mathbf{x} \in \mathbb{A}^{N\times 1}$ denote the vector of (QAM) information symbols in the discrete affine Fourier domain, where $\mathbb{A} = \{a_0, \cdots, a_{Q-1} \}$ represents the QAM alphabet. Note that $\mathbb{A} \subset \mathbb{Z}[j]$ where $\mathbb{Z}[j]$ denotes the number field whose elements have the form $z_r + jz_i$, with integer $z_r$ and $z_i$. The modulated signal can be written as 
\begin{equation}
\label{eq:mod}
    s_n = \sum_{m = 0}^{N-1}x_m\phi_n(m), \quad n= 0 , \cdots, N-1
\end{equation}
where $\phi_n(m) = e^{j2\pi (c_1n^2 + c_2m^2 + nm/N)}/\sqrt{N}$. In matrix form, \eqref{eq:mod} becomes
\begin{equation}
    {\mathbf s} = {\mathbf \Lambda}_{c_1}^H{\mathbf F}^H{\mathbf \Lambda}_{c_2}^H\mathbf{x}.
\end{equation}

Similarly to OFDM, the proposed scheme needs some kind of prefix to deal with multipath propagation and to make the channel seemingly lie in a periodic domain. However, due to different signal periodicity, a {\emph{chirp-periodic}} prefix (CPP) has to be used instead of an OFDM cyclic prefix (CP). Indeed, an $L$-long prefix occupying the positions of the negative-index time-domain samples should be transmitted, where $L$ is any integer greater than or equal to the value in samples of the maximum delay spread of the wireless channel. With the periodicity defined in \eqref{chrip_per}, the prefix is
\begin{equation}
    s_{n} = s_{N+n}e^{-j2\pi c_1(N^2+2Nn)},\quad n = -L, \cdots, -1.
\end{equation}
Note that a CPP is simply a CP whenever $2Nc_1$ is an integer value and $N$ is even.
\subsection{Channel}
\label{sec:channel}
After parallel to serial conversion and transmission over the channel, the received samples are 
\begin{equation}
    r_n = \sum_{l = 0}^{\infty}s_{n-l}g_n(l)\label{r_n} + w_n,
\end{equation}
where $w_n\sim\mathcal{CN}\left(0,N_0\right)$ is additive Gaussian noise and
\begin{equation}
    g_n(l) = \sum _{i=1}^{P} h_{i}e^{-j2\pi f_in}\delta(l - l_i)
\end{equation}
is the impulse response of channel at time $n$ and delay $l$, where $P\geq1$ is the number of paths, $\delta(\cdot)$ is the Dirac delta function, and $h_i, f_i$ and $l_i$ are the complex gain, Doppler shift (in digital frequencies), and the integer delay associated with the $i$-th path, respectively. Note that this model is general and covers the case where each delay tap can have a Doppler frequency \emph{spread} by simply allowing for different paths $i,j\in\{1,\ldots,P\}$ to have the same delay $l_i=l_j$ while satisfying $f_i\neq f_j$.
We define $\nu_i \triangleq{} N f_i = \alpha_i +a_i $, where $\nu_i\in\left[-\nu_{\max},\nu_{\max}\right]$ is the Doppler shift normalized with respect to the subcarrier spacing, $\alpha_i\in\left[-\alpha_{\max},\alpha_{\max}\right]$ is its integer part while $a_i$ is  the fractional part satisfying $\frac{-1}{2} < a_i \leq \frac{1}{2}$. For the sake of simplifying the diversity analysis, we assume that the fractional parts $a_i$ are zero, which is reasonable since $a_i$ can be neglected for large values of $N$. In addition, we assume that the maximum delay of the channel satisfies $l_{\max}\triangleq\max (l_i)  < N$, and that the CPP length is greater than $l_{\max}-1$ ($L > l_{\max}-1$).

After discarding the CPP, we can write \eqref{r_n} in the matrix form
\begin{equation}
    {\mathbf r} = {\mathbf H}\mathbf{s} + \mathbf{w}
\end{equation}
where $\mathbf{w}\sim\mathcal{CN}\left(\mathbf{0},N_0\mathbf{I}\right)$ and $\mathbf{H}$ is the $N\times N$ matrix 
\begin{equation}
{\mathbf H}= \sum _{i=1}^{P} h_{i} {\mathbf {\Gamma }}_{\mathrm{CPP}_{i}} {\mathbf {\Delta }}_{f_{i}} {\mathbf{\Pi }}^{l_{i}}
\end{equation}
where $\mathbf{\Pi }$ is the permutation matrix
\begin{align} 
{\bf {\Pi }} & = {\left[\begin{array}{cccc} 0 & \cdots & 0 & 1\\ 1 & \cdots & 0 & 0\\ \vdots & \ddots & \ddots & \vdots \\ 0 & \cdots & 1 & 0 \end{array}\right]_{N \times\; N}},
\end{align}
${\mathbf {\Delta }}_{f_{i}}$ is the $N\times N$ diagonal matrix
\begin{equation}
\Delta_{f_i} ={\mathrm{diag}}(e^{-j2\pi f_in}, n = 0, 1, \, \ldots,\, N-1)
\end{equation}
and ${\mathbf {\Gamma }}_{\mathrm{CPP}_{i}} $ is a $N\times N$ diagonal matrix 
\begin{equation}
    \begin{multlined}
    \mathbf {\Gamma }_{\mathrm{CPP}_{i}} = \\ {\mathrm{diag}}(\begin{cases}e^{-j2\pi c_1(N^{2}-2N(l_i - n))}&n < l_i\\ 1&n\geq l_i\end{cases}, n = 0,\, \ldots,\, N-1).
    \end{multlined}
    \label{eq:CP_matrix}
\end{equation}
We can see from \eqref{eq:CP_matrix} that whenever $2Nc_1$ is an integer and $N$ is even, $\mathbf {\Gamma }_{\mathrm{CPP}_{i}}=\mathbf {I}$.

\subsection{Demodulation}
At the receiver side, the DAFT domain output symbols are obtained by
\begin{equation}
    y_m = \sum_{n = 0}^{N-1}r_n\phi_n^*(m).
    \label{eq:y_received}
\end{equation}
In matrix representation, the output can be written as
\begin{equation}
\begin{aligned}
    {\mathbf y} =& {\mathbf \Lambda}_{c_2}{\mathbf F}{\mathbf \Lambda}_{c_1}\mathbf{r}\\
    =& \sum _{i=1}^{P} h_{i} {\mathbf \Lambda}_{c_2}{\mathbf F}{\mathbf \Lambda}_{c_1}{\mathbf {\Gamma }}_{\mathrm{CPP}_{i}} {\mathbf {\Delta }}_{f_{i}} {\mathbf{\Pi }}^{l_{i}} {\mathbf \Lambda}_{c_1}^H{\mathbf F}^H{\mathbf \Lambda}_{c_2}^H{\mathbf x}+ \widetilde {\mathbf w}\\
    =& {\mathbf H}_{\mathrm{eff}} {\mathbf x}+ \widetilde {\mathbf w}\label{eq_rec}
\end{aligned}
\end{equation}
where  ${\mathbf H}_{\mathrm{eff}} \triangleq {\mathbf \Lambda}_{c_2}{\mathbf F}{\mathbf \Lambda}_{c_1}{\mathbf H}{\mathbf \Lambda}_{c_1}^H{\mathbf F}^H{\mathbf \Lambda}_{c_2}^H$ and $ \widetilde {\mathbf w}$ is ${\mathbf \Lambda}_{c_2}{\mathbf F}{\mathbf \Lambda}_{c_1}\mathbf{w}$. Since ${\mathbf \Lambda}_{c_2}{\mathbf F}{\mathbf \Lambda}_{c_1}$ is a unitary matrix, $\widetilde {\mathbf w}$ and ${\mathbf w}$ have the same covariance.

\subsection{Input-Output Relation}
Considering ${\mathbf H}_{\mathrm{eff}} = \sum_{i = 1}^{P}h_i\mathbf{H}_i$, \eqref{eq_rec} can be rewritten as
\begin{equation}
    \mathbf{y} = \sum_{i = 1}^{P}h_i\mathbf{H}_i\mathbf{x} + \Tilde{\mathbf{w}}\label{eq:y_output}
\end{equation}
with
\begin{equation}
    \mathbf{H}_i \triangleq {\mathbf \Lambda}_{c_2}{\mathbf F}\underbrace{{\mathbf \Lambda}_{c_1}{\mathbf {\Gamma }}_{\mathrm{CPP}_{i}} {\mathbf {\Delta }}_{f_{i}} {\mathbf{\Pi }}^{l_{i}}{\mathbf \Lambda}_{c_1}^H}_{\mathbf{A}_i}{\mathbf F}^H{\mathbf \Lambda}_{c_2}^H.\label{H_i}
\end{equation}
It can be shown that the element of $\mathbf{A}_i$ at row $n$ and column $(n-l_i)_N$ is
\begin{equation}
\begin{aligned}
    \mathbf{A}_i(n, (n-l_i)_N) = \frac{1}{h_i}g_n(l_i)e^{j2\pi c_1(l_i^2-2nl_i)}
\end{aligned}
\end{equation}
where $(\cdot)_N$ is the modulo $N$ operation.

From \eqref{H_i}, the element of $\mathbf{H}_i$ at row $p$ and column $q$ writes as
{\footnotesize	
\begin{equation}
    \begin{aligned}
    &\mathbf{H}_{i}(p, q) = \frac{e^{-j2\pi c_2(p^2 - q^2)}}{N}\sum_{n = 0}^{N-1} e^{-j\frac{2\pi}{N}(pn - q(n-l_i)_N)}\mathbf{A}_i(n, (n-l_i)_N)\\
    &= \frac{1}{N}e^{j\frac{2\pi}{N}(Nc_1l_i^2 - ql_i + Nc_2(q^2 - p^2))}
    \sum_{n = 0}^{N-1} e^{-j\frac{2\pi}{N}((p-q + \nu_i + 2Nc_1 l_i)n)}.\label{eq:H_i_general}
\end{aligned}
\end{equation}
}

As mentioned before, $\nu_i$ is assumed to be integer valued for all $i\in\{1,\ldots,P\}$, i.e., $\nu_i=\alpha_i$. Moreover, if $c_1$ is chosen such that $2Nc_1 l_i$ is an integer,
\eqref{eq:H_i_general} writes as
\begin{equation}
    \mathbf{H}_i(p, q) = 
    \begin{cases} e^{j\frac{2\pi}{N}(Nc_1l_i^2 - ql_i + Nc_2(q^2 - p^2))} & q = (p+\mathrm{loc}_i)_N  \\
    0 & {\text{otherwise}}
    \end{cases},
    \label{eq:Hi_p_q}
\end{equation}
where $\mathrm{loc}_i \triangleq \alpha_i + 2Nc_1 l_i$.
Hence, there is only one non-zero element in each row of $\mathbf{H}_i$ and its location in the $p$-th row is $(p+\mathrm{loc}_i)_N$. The input-output relation in \eqref{eq:y_output} becomes
\begin{equation}
    y_p = \sum_{i = 1}^{P}h_ie^{j\frac{2\pi}{N}(Nc_1{l_i^2} - ql_i + Nc_2(q^2 - p^2))}x_{q} + \Tilde{\mathbf{w}},  0 \leq p \leq N-1
    \label{eq:y_output_special}
\end{equation}
where $q = (p+\mathrm{loc}_i)_N$.
\section{AFDM Parameters}

The performance of DAFT-based modulation schemes critically depends on the choice of parameters $c_1$ and $c_2$. OCDM uses $c_1 = c_2 = \frac{1}{2N}$. As shown later, this choice provides better diversity performance than OFDM, but fails to achieve full diversity in LTV channels. In DAFT-OFDM, $c_2 = 0$ while $c_1$ is adapted to the delay-Doppler channel profile to minimize ICI. While this choice simplifies detection, we show that it fails to achieve full diversity. In the proposed AFDM, we set $c_1$ and $c_2$ in a way that the DAFT domain impulse response constitutes a full delay-Doppler representation of the channel. This choice allows AFDM to achieve full diversity in LTV channels, as shown in Section \ref{sec:diversity}.

In order for the DAFT domain impulse response to constitute a full delay-Doppler representation of the channel, the unique non-zero entry in each row of $\mathbf{H}_i$ for each path $i\in\{1,\ldots,P\}$ should not coincide with the position of the unique non-zero entry of the same row of $\mathbf{H}_j$ for any $j\in\{1,\ldots,P\}$ such that $j\neq i$. Referring to \eqref{eq:Hi_p_q} shows that the location of each path depends on its delay-Doppler information and AFDM parameters. For the integer Doppler shift case, $\mathrm{loc}_i$ is in the following range
\begin{equation}
    -\alpha_{\rm{max}} + 2Nc_1l_i \leq \mathrm{loc}_i \leq 
    \alpha_{\rm{max}} + 2Nc_1l_i.
    \label{eq:range_integer}
\end{equation}
Therefore, for the positions of the non-zero entries of $\mathbf{H}_i$ and $\mathbf{H}_j$ to not overlap, the intersection of the corresponding ranges of $\mathrm{loc}_i$ and $\mathrm{loc}_j$ should be empty, i.e, 
\begin{equation}
\begin{aligned}
    &\{-\alpha_{\rm{max}} + 2Nc_1l_i, ..., \alpha_{\rm{max}} + 2Nc_1l_i\} \cap\\
    &\{-\alpha_{\rm{max}} + 2Nc_1l_j, ..., \alpha_{\rm{max}} + 2Nc_1l_j\}= \emptyset.
    \label{eq:intersec_int}
\end{aligned}
\end{equation}
If two paths have the same delays ($l_i = l_j$) but different Doppler shifts, then they always occupy two distinct positions in the DAFT domain. For the paths with different delays ($l_i \neq l_j$) assuming $l_j>l_i$, satisfying \eqref{eq:intersec_int} is equivalent to the constraint
\begin{equation}
    2Nc_1 > \frac{2\alpha_{\rm{max}}}{l_j - l_i}.
\end{equation}
If there is no sparsity in the time-domain impulse response of the channel, then the minimum value of $l_j - l_i $ is one and $c_1$ should satisfy
\begin{equation}
\label{eq:opt_c1}
    \boxed{c_1 = \frac{2\alpha_{\max} + 1}{2N}.}
\end{equation}
\begin{figure}
  \centering
  \includegraphics[scale=.25]{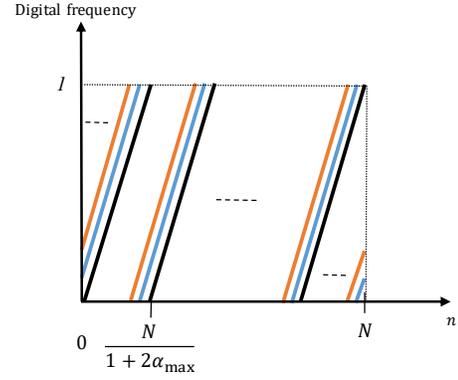}
  \vspace{-3mm}
  \caption{Time-frequency representation of each chirp in AFDM using its $c_1$ given in \eqref{eq:opt_c1}}
    \vspace{-2mm}

  \label{fig:TF_rep}
\end{figure}
With this $c_1$, the time-frequency representation of chirps is shown in Fig. \ref{fig:TF_rep}. Moreover, the only remaining condition for the DAFT-domain impulse response to constitute a full delay-Doppler representation of the channel is to ensure that the non-zero entries of any two matrices $\mathbf{H}_{i_{\min}}$ and $\mathbf{H}_{i_{\max}}$ corresponding to paths $i_{\min}$ and $i_{\max}$ with delays $l_{i_{\min}}\triangleq\min_{i=1\cdots P}l_i$ and $l_{i_{\max}}\triangleq\max_{i=1\cdots P}l_i$ respectively do not overlap due to the modular operation in \eqref{eq:Hi_p_q}. This overlapping never occurs if $2\alpha_{\rm{max}}l_{\rm{max}} + 2\alpha_{\rm{max}} + l_{\rm{max}}< N$. Since the wireless channels are typically underspread, i.e., $l_{\rm{max}} \ll N$ and $\alpha_{\rm{max}} \ll N$, this condition can be satisfied even with moderate values of $N$. 
\begin{figure}
  \centering
  \includegraphics[scale=.55]{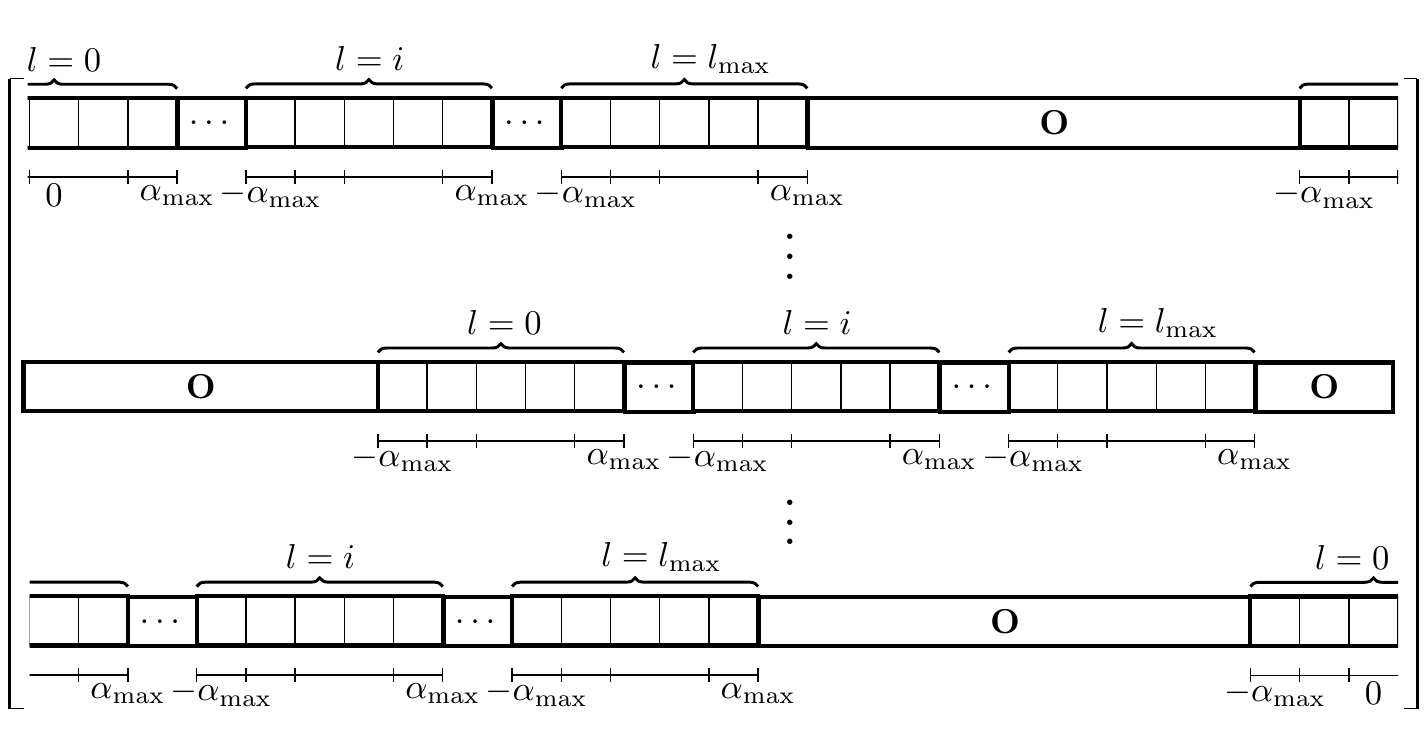}
  \vspace{-4mm}
  \caption{Structure of $\textbf{H}_{\mathrm{eff}}$ in AFDM}
  \vspace{-4mm}
  \label{fig:Channelpic}
\end{figure}
With this parameter setting, channel paths with different delay values or different Doppler frequency shifts get separated in the DAFT domain, resulting in $\textbf{H}_{\mathrm{eff}}$ having the structure shown in Fig. \ref{fig:Channelpic}. Thus, we get a delay-Doppler representation of the channel in the DAFT domain since the delay-Doppler profile can be determined from the positions of the non-zero entries in any row of $\textbf{H}_{\mathrm{eff}}$. This feature can neither be obtained by DAFT-OFDM (since its conceptual target is making the effective channel matrix as close to being diagonal as possible to reduce ICI), nor with OCDM (since setting $c_1 =\frac{1}{2N}$, there might exist two paths $i\neq j$ such that the non-zero entries of $\mathbf{H}_i$ and $\mathbf{H}_j$ coincide under some delay-Doppler profiles of the channel). We next show that this unique feature of AFDM translates into its optimality in terms of achievable diversity order in LTV channels.

\section{Diversity Analysis of AFDM}
\label{sec:diversity}
We start by rewriting \eqref{eq_rec} as 
\begin{equation}
{\mathbf y}= \sum _{i=1}^{P} h_{i} {\mathbf H}_{i}{\mathbf x}+ \widetilde {\mathbf w} = {\mathbf{\Phi }}({\mathbf {x}}) {\mathbf h} + \widetilde {\mathbf w}
\end{equation}
where $\mathbf{h}=[h_1,h_2, \ldots,h_P]^T$ is a $P\times 1$ vector and ${\boldsymbol{\Phi}}({\mathbf{x}})$ is the $N \times P$ concatenated matrix 
\begin{equation}
    {\boldsymbol{\Phi}}({\mathbf x}) = [\mathbf{H}_1\mathbf{x} \mid \ldots \mid \mathbf{H}_P\mathbf{x}].
    \label{eq:phi_x}
\end{equation}
The conditional pairwise error probability (PEP) between $\mathbf{x}_m$ and $\mathbf{x}_n$, i.e., transmitting symbol $\mathbf{x}_m$ and deciding in favor of $\mathbf{x}_n$ at the receiver, is given by
\begin{equation}
P(\mathbf {x}_{m}\rightarrow \mathbf {x}_{n}|\mathbf {h},\mathbf {x}_{m})=Q \left ({\sqrt {\frac {\|{\mathbf{\Phi }}(\mathbf {x}_{m}-\mathbf {x}_{n})\mathbf {h}\|^{2}}{2N_{0}}} }\right)\!.
\end{equation}
By averaging over the channel realizations, the PEP becomes
\begin{equation}
P(\mathbf {x}_{m}\rightarrow \mathbf {x}_{n})=\mathbb {E}_\mathbf{h} \left [{ Q \left ({\sqrt {\frac {~\| \mathbf{\Phi}(\boldsymbol{\delta}^{(m,n)})\mathbf {h}\|^{2}}{2N_0}}\, }\right) }\right]\!
\end{equation}
where $\boldsymbol{\delta}^{(m,n)} \triangleq \mathbf{x}_m - \mathbf{x}_n$. Assuming $h_i$s are i.i.d and distributed as $\mathcal {CN}(0,1/P)$, it can be shown that \cite{tse2005fundamentals}
\begin{equation}
P(\mathbf {x}_{m}\rightarrow \mathbf {x}_{n}) \leq \prod \limits _{l=1}^{r}\frac {1}{1+\,\,\dfrac {\lambda _{l}^{2}}{4PN_0}}\label{upper_bound}
\end{equation}
where $r$ is the rank of matrix $\boldsymbol{\Phi}(\boldsymbol{\delta}^{(m,n)})$ and $\lambda_l$ is its $l$th singular value. Moreover, \eqref{upper_bound} implies that at high values of the signal-to-noise ratio $\mathrm{SNR}\triangleq\frac{1}{N_0}$, the overall bit error rate (BER) is dominated by the PEP with the minimum value of $r$, for all pairs $(m, n), m\neq n$. Hence, the diversity order of AFDM is given by
\begin{equation}
\label{eq:ord_div}
\begin{aligned}
\rho  \triangleq \min _{m,n ~m\neq n}~\text {rank}(\boldsymbol{\Phi }(\boldsymbol{\delta})^{(m, n)})
\leq P\:.
\end{aligned}
\end{equation}
For exposition convenience, we drop $(m, n)$ from ${\delta}^{(m, n)}$. 
The following theorem states that AFDM achieves full diversity  provided that $N$ is large enough.
\begin{theorem}
	\label{theo:main}
For a linear time-varying channel with a maximum delay $l_{\rm{max}}$ and maximum normalized Doppler shift $\alpha_{\rm{max}}$, AFDM with $c_1$ satisfying \eqref{eq:opt_c1} achieves full diversity i.e., $\rho=P$ if 
\begin{align}
    2\alpha_{\rm{max}} + l_{\rm{max}} + 2\alpha_{\rm{max}}l_{\rm{max}} < N \label{eq:cond_1}
\end{align}
\end{theorem}
\begin{proof}
To prove Theorem \ref{theo:main}, we show that when \eqref{eq:opt_c1} and \eqref{eq:cond_1} hold, there exist values of $c_2$ such that the rank of $\mathbf{\Phi}(\boldsymbol{\delta})$ should be $P$ i.e., such that the $P$ columns of $\mathbf{\Phi}(\boldsymbol{\delta})$ are linearly independent. For the sake of clarity of presentation under the constraint on the number of pages, we give the proof in the special case of a two-path channel ($P = 2$). The proof in the general case follows the same steps but requires a longer development. Assuming $P=2$,   $\boldsymbol{\Phi}(\boldsymbol\delta)$ is written in \eqref{eq:phi_delta} at the bottom of the page.
\begin{figure*}[b]
\begin{equation}
{\boldsymbol{\Phi}(\boldsymbol\delta) = 
\left[{\begin{array}{cc}
        \mathbf{H}_{\rm{eff}}(0, \mathrm{loc}_1)\delta_{\mathrm{loc}_1}
        & \mathbf{H}_{\rm{eff}}(0, \mathrm{loc}_2)\delta_{\mathrm{loc}_2}\\
        \vdots &\vdots\\ 
\mathbf{H}_{\rm{eff}}(N-1, (\mathrm{loc}_1+N-1)_N)\delta_{(\mathrm{loc}_1+N-1)_N}     
        &\mathbf{H}_{\rm{eff}}(N-1, (\mathrm{loc}_2+N-1)_N)\delta_{(\mathrm{loc}_2+N-1)_N}        
\end{array}}\right].
}\label{eq:phi_delta}
\end{equation}
\end{figure*}
We now show that there exists $c_2$ such that the rank of matrix $\boldsymbol{\Phi}(\boldsymbol{\delta})$ is two. Given that we exclude the
case $\boldsymbol{\delta} = 0$ ($\textbf{x}_m \neq \textbf{x}_n$), there is at least one non-zero entry $\delta_z\neq 0$ of $\boldsymbol{\delta}$ for some $z \in \{0, \cdots, N-1\}$. Assume without loss of generality that $ \mathrm{loc}_2<z<N-\alpha_{\max}$.
Matrix $\boldsymbol{\Phi}(\boldsymbol{\delta})$ is of rank two if the following matrix constructed with two rows of $\boldsymbol{\Phi}(\boldsymbol{\delta})$ is full rank
\begin{equation}
{\footnotesize
\left[{\begin{array}{cc}
        {\mathbf{H}_{\rm{eff}}(z-\mathrm{loc}_2, z - \mathrm{loc}_d)\delta_{z - \mathrm{loc}_d}}
        & \mathbf{H}_{\rm{eff}}(z-\mathrm{loc}_2, z)\delta_{z}\\
\mathbf{H}_{\rm{eff}}(z-\mathrm{loc}_1, z)\delta_{z}     
        &\mathbf{H}_{\rm{eff}}(z-\mathrm{loc}_1, z + \mathrm{loc}_d)\delta_{z + \mathrm{loc}_d}        
\end{array}}\right]
}
\end{equation}
where $\mathrm{loc}_d = \mathrm{loc}_2 - \mathrm{loc}_1$. Since a matrix $\left[{\begin{array}{cc} t_1 & t_2 \\ t_3 & t_4 \end{array}}\right]$ is of full rank if $t_1t_4 \neq t_2t_3$, we need $c_2$ to guarantee that the following inequality holds 
\begin{equation}
    \delta_{z}^2 {\neq} e^{j\frac{2\pi}{N}(l_2- l_1)\mathrm{loc}_d}e^{j4\pi c_2\mathrm{loc}_d^2}\delta_{z-\mathrm{loc}_d}\delta_{z+\mathrm{loc}_d}
    \label{eq:delta_0_2}
\end{equation}
Note that $\boldsymbol{\delta} \in \mathbb{Z}[j]^{N\times 1}$. Therefore, since $\delta_{z} \in \mathbb{Z}[j]$ then the right-hand side of \eqref{eq:delta_0_2} should not be in $\mathbb{Z}[j]$ for the inequality to hold. Now, since $\mathrm{loc}_d^2\neq 0$, setting $c_2$ to be either an arbitrary irrational number or a rational number sufficiently smaller than $\frac{1}{2N}$ guarantees that the right-hand side of \eqref{eq:delta_0_2} is not in $\mathbb{Z}[j]$ and thus makes the inequality in \eqref{eq:delta_0_2} hold. Therefore,  $\boldsymbol{\Phi}(\boldsymbol{\delta})$ is a full-rank matrix and AFDM achieves the full diversity of the channel.
\end{proof}

\section{Simulation Results}
\label{sec:simus}
In this section, we compare the BER performance of AFDM to that of DAFT-OFDM\cite{erseghe2005multicarrier}, OCDM\cite{ouyang2016orthogonal} and OTFS\cite{hadani2017orthogonal}. 
In all simulations, BER values are obtained using $10^6$ different channel realizations with complex gains $h_i$ generated as independent complex Gaussian random variables with zero mean and $1/P$ variance. 
\begin{figure}
    \centering
    \includegraphics[scale=.55]{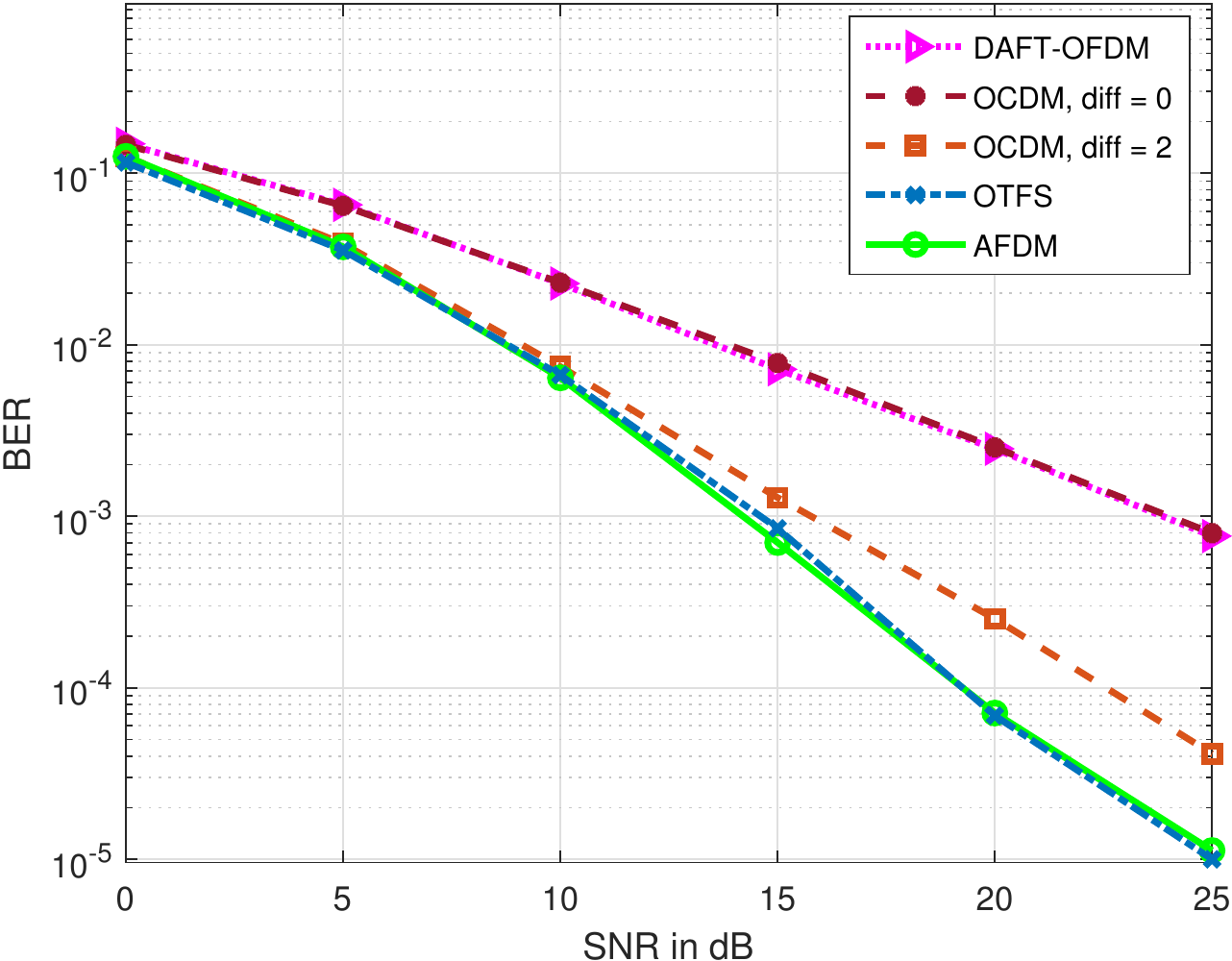}
    \caption{BER performance of  DAFT-OFDM, OCDM, OTFS and AFDM using BPSK in a two-path LTV channel with $l_{
    max} = 1$ and $\alpha_{\max} = 1$ for $N = 8$,  $N_{\rm{OTFS}} = 4$ and $M_{\rm{OTFS}} =2$ using ML detection.}
    \label{fig_L4_c2_1N}
\end{figure}

Fig. \ref{fig_L4_c2_1N} shows the BER performance of the four schemes with $N = 8$ and $N_{\rm{OTFS}} = 4$, $M_{\rm{OTFS}} = 2$ \footnote{$N_{\rm{OTFS}}$ and $M_{\rm{OTFS}}$ are used to discretize the time-frequency signal plane and delay-Doppler plane to $M_{\rm{OTFS}}\times N_{\rm{OTFS}}$ grids. More details can be found in \cite{hadani2017orthogonal}} for OTFS, in a two-path LTV channel with different delay-Doppler profiles, using BPSK symbols and maximum likelihood (ML) detection. ML is employed for the purposes of diversity order comparison; different detection methods can, of course, be used in practice. 
Notation \textit{diff} designates the distance between the location of the two non-zero elements in each row of matrix $\mathbf{H}_{\mathrm{eff}}$. It is observed that DAFT-OFDM has always diversity order one, since it has always one non-zero element in each row of its associated $\mathbf{H}_{\mathrm{eff}}$. The performance of OCDM depends on \textit{diff}. When $\textit{diff} = 0$, OCDM performs poorly and has the same diversity (one) as DAFT-OFDM, mainly due to the possible destructive addition of the two overlapping paths in that case. Even with two non-zero elements in each row of its $\mathbf{H}_{\mathrm{eff}}$ ($\textit{diff} = 2$), full diversity cannot be achieved because in OCDM $c_2 = \frac{1}{2N}$, which does not guarantee that ${\boldsymbol{\Phi}}({\boldsymbol{\delta}})$ is full rank as shown in the proof of Theorem \ref{theo:main}. In sharp contrast, AFDM achieves full diversity, mainly due to the path separation by tuning $c_1$ and setting $c_2$ to be an arbitrary irrational number or a rational number sufficiently smaller than $\frac{1}{2N}$. Expectedly, AFDM has the same BER performance as OTFS.
\begin{figure}
    \centering
    \includegraphics[scale=.55]{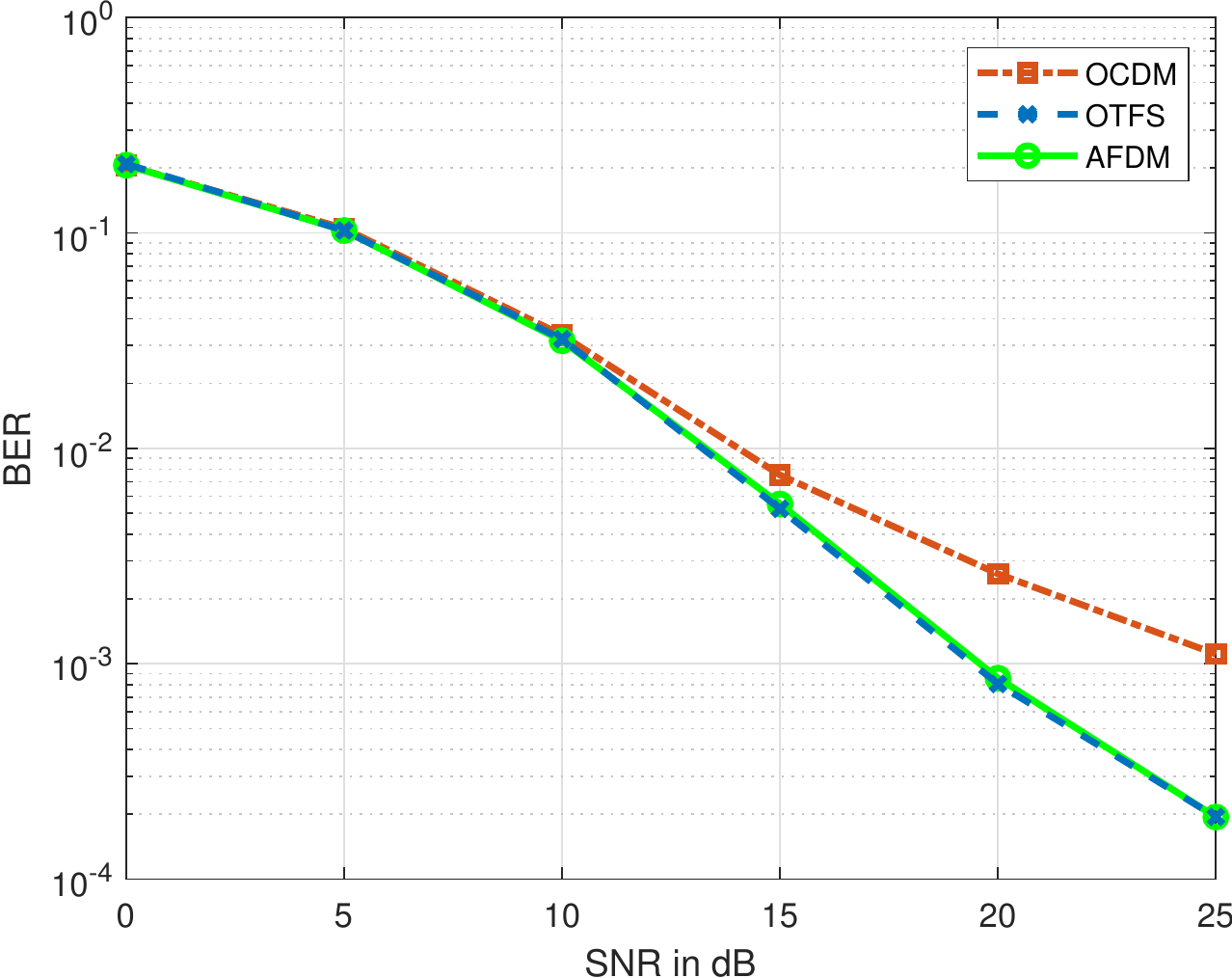}
    \caption{BER performance of OCDM, OTFS and AFDM using QPSK in a 21-path LTV channel with $l_{\max} = 2$ and $\alpha_{\max} = 3$ for $N = 64$,  $N_{\rm{OTFS}} = 8$ and $M_{\rm{OTFS}} =8$ using MMSE detection.}
    \label{fig:P-4_LMMSE}
\end{figure}

Fig. \ref{fig:P-4_LMMSE} shows the performance of OCDM, AFDM with $N = 64$ and OTFS with $N_{\rm{OTFS}} = 8$ and $M_{\rm{OTFS}} = 8$, using QPSK symbols and minimum mean square error (MMSE) detection in a 21-path LTV channel with $l_{\max} = 2$ and $\alpha_{\max} = 3$. For each delay tap, there are 7 paths with different Doppler shifts from -3 to 3. We observe that OCDM exhibits the worst performance due to the destructive addition of overlapping paths in the effective channel matrix. In contrast, AFDM does not experience this destructive effect thanks to its path separation. Moreover, we can see that AFDM and OTFS have again the same performance in terms of BER. As mentioned before, although OTFS achieves the same effective full diversity as AFDM, it suffers from excessive pilot overhead as each pilot symbol needs at least $(2l_{\max} + 1)(4\alpha_{\max} + 1)-1$ 
guard symbols to avoid data/pilot interference \cite{raviteja2019embedded} due to the 2D structure of its underlying transform. As for AFDM, only
$(2l_{\max}+2)(2\alpha_{\max}+1) - 2$ guard symbols are needed as can be verified from the structure of $\mathbf{H}_{\mathrm{eff}}$ given in Fig. \ref{fig:Channelpic}. This key difference, which will presented in detail in a longer version of this paper, entails a significant throughput gap in favor of AFDM. The throughput gap increases with the number of required orthogonal pilot transmissions.

\section{Conclusions}
In this paper, we have proposed AFDM, a new waveform based on multiple discrete-time orthogonal chirp signals. Chirps are generated using the discrete affine Fourier transform that is characterized by two parameters. We have derived its input-output signal relation on doubly dispersive channels and have set the AFDM parameters so that the DAFT domain channel impulse response constitutes a full representation of its delay-Doppler profile. Our analytical results have shown that AFDM always achieves full diversity in doubly dispersive channels. AFDM is a promising next generation waveform for high mobility communications, which outperforms OFDM and other DAFT-based multicarrier schemes, while having advantages over OTFS in terms of pilot and user multiplexing overhead.

\bibliographystyle{IEEEtran}
\bibliography{IEEEabrv,Citations}

\begin{thebibliography}{10}
\providecommand{\url}[1]{#1}
\csname url@samestyle\endcsname
\providecommand{\newblock}{\relax}
\providecommand{\bibinfo}[2]{#2}
\providecommand{\BIBentrySTDinterwordspacing}{\spaceskip=0pt\relax}
\providecommand{\BIBentryALTinterwordstretchfactor}{4}
\providecommand{\BIBentryALTinterwordspacing}{\spaceskip=\fontdimen2\font plus
\BIBentryALTinterwordstretchfactor\fontdimen3\font minus
  \fontdimen4\font\relax}
\providecommand{\BIBforeignlanguage}[2]{{%
\expandafter\ifx\csname l@#1\endcsname\relax
\typeout{** WARNING: IEEEtran.bst: No hyphenation pattern has been}%
\typeout{** loaded for the language `#1'. Using the pattern for}%
\typeout{** the default language instead.}%
\else
\language=\csname l@#1\endcsname
\fi
#2}}
\providecommand{\BIBdecl}{\relax}
\BIBdecl

\bibitem{erseghe2005multicarrier}
T.~Erseghe, N.~Laurenti, and V.~Cellini, ``A multicarrier architecture based
  upon the affine {Fourier} transform,'' \emph{IEEE Trans. on Commun.},
  vol.~53, no.~5, pp. 853--862, May 2005.

\bibitem{martone2001multicarrier}
M.~Martone, ``A multicarrier system based on the fractional {Fourier} transform
  for time-frequency-selective channels,'' \emph{IEEE Trans. on Commun.},
  vol.~49, no.~6, pp. 1011--1020, Jun. 2001.

\bibitem{ouyang2016orthogonal}
X.~Ouyang and J.~Zhao, ``Orthogonal chirp division multiplexing,'' \emph{IEEE
  Trans. on Commun.}, vol.~64, no.~9, pp. 3946--3957, Sept. 2016.

\bibitem{hadani2017orthogonal}
R.~Hadani, S.~Rakib, M.~Tsatsanis, A.~Monk, A.~J. Goldsmith, A.~F. Molisch, and
  R.~Calderbank, ``Orthogonal time frequency space modulation,'' in \emph{2017
  IEEE Wireless Communications and Networking Conference (WCNC)}.\hskip 1em
  plus 0.5em minus 0.4em\relax IEEE, 2017, pp. 1--6.

\bibitem{surabhi2019diversity}
G.~Surabhi, R.~M. Augustine, and A.~Chockalingam, ``On the diversity of uncoded
  {OTFS} modulation in doubly-dispersive channels,'' \emph{IEEE Trans. on
  Wireless Communications}, vol.~18, no.~6, pp. 3049--3063, 2019.

\bibitem{raviteja2019effective}
P.~Raviteja, Y.~Hong, E.~Viterbo, and E.~Biglieri, ``Effective diversity of
  {OTFS} modulation,'' \emph{IEEE Wireless Communications Letters}, Feb. 2019.

\bibitem{raviteja2019embedded}
P.~Raviteja, K.~T. Phan, and Y.~Hong, ``Embedded pilot-aided channel estimation
  for {OTFS} in delay--doppler channels,'' \emph{IEEE Trans. on Vehicular
  Technology}, vol.~68, no.~5, pp. 4906--4917, 2019.

\bibitem{LCT}
M.~Moshinsky and C.~Quesne, ``Linear canonical transformations and their
  unitary representations,'' \emph{Journal of Mathematical Physics}, vol.~12,
  no.~8, pp. 1772--1780, Aug. 1971.

\bibitem{pei2001relations}
S.-C. Pei and J.-J. Ding, ``Relations between fractional operations and
  time-frequency distributions, and their applications,'' \emph{IEEE Trans. on
  Sig. Proc.}, vol.~49, no.~8, pp. 1638--1655, Aug. 2001.

\bibitem{SAFT}
S.~Abe and J.~T. Sheridan, ``Generalization of the fractional {Fourier}
  transformation to an arbitrary linear lossless transformation an operator
  approach,'' \emph{Journal of Physics A: Mathematical and General}, vol.~27,
  no.~12, pp. 4179--4187, June 1994.

\bibitem{candan2000discrete}
C.~Candan, M.~A. Kutay, and H.~M. Ozaktas, ``The discrete fractional {Fourier}
  transform,'' \emph{IEEE Trans. on Sig. Proc.}, vol.~48, no.~5, pp.
  1329--1337, May 2000.

\bibitem{pei2000closed}
S.-C. Pei and J.-J. Ding, ``Closed-form discrete fractional and affine
  {Fourier} transforms,'' \emph{IEEE Trans. on Sig. Proc.}, vol.~48, no.~5, pp.
  1338--1353, May 2000.

\bibitem{tse2005fundamentals}
D.~Tse and P.~Viswanath, \emph{Fundamentals of wireless communication}.\hskip
  1em plus 0.5em minus 0.4em\relax Cambridge University Press, 2005.

\end{thebibliography}
\end{document}